\documentclass{siamonline171218}

\usepackage{amsfonts,amssymb,mathtools,bm}
\usepackage[OT1]{fontenc}
\usepackage{cleveref}

\newtheorem{example}{Example}
\newtheorem{remark}{Remark}

\def\eqns#1{\begin{equation*}#1\end{equation*}}
\def\eqnl#1#2{\begin{equation}\label{#1}#2\end{equation}}
\def\eqnsa#1{\begin{subequations}\begin{align*}#1\end{align*}\end{subequations}}

\def\eqnmla#1#2{\begin{subequations}\label{#1}\begin{align}#2\end{align}\end{subequations}}

\def\one{\mathbf{1}}
\def\boE{\mathbf{E}}
\def\boF{\mathbf{F}}
\def\bsf{\bm{f}}
\def\bsF{\bm{F}}
\def\bsI{\bm{I}}
\def\bsK{\bm{K}}
\def\boL{\mathbf{L}}
\def\bsO{\bm{O}}
\def\bsP{\bm{P}}
\def\bsQ{\bm{Q}}
\def\bsR{\bm{R}}
\def\bsx{\bm{x}}
\def\bsy{\bm{y}}
\def\boX{\mathbf{X}}
\def\boY{\mathbf{Y}}
\def\bsvphi{\bm{\varphi}}

\def\calB{\mathcal{B}}
\def\calE{\mathcal{E}}
\def\calF{\mathcal{F}}
\def\calN{\mathcal{N}}

\def\bbN{\mathbb{N}}
\def\bbP{\mathbb{P}}
\def\bbR{\mathbb{R}}
\def\bbT{\mathbb{T}}

\def\d{\mathrm{d}}
\def\r{\mathrm{r}}
\def\u{\mathrm{u}}

\def\frakP{\mathfrak{P}}

\def\defeq{\doteq}
\def\cev#1{\reflectbox{\ensuremath{\vec{\reflectbox{\ensuremath{#1}}}}}}
\def\ind#1{\one_{#1}}
\def\scl{\dagger}
\def\given{\,|\,}
\def\AND{\qquad\text{and}\qquad}
\def\up{:}
\def\down{:}

\def\st{\,:\,}

\def\oP{\bar{P}}

\def\oQ{\bar{Q}}
\def\oR{\bar{R}}
\def\ninf#1{\|#1\|_{\infty}}

\def\tf{\varphi}
\def\bstf{\bsvphi}

\title{Smoothing and filtering with a class of outer measures}
\author{Jeremie Houssineau and Adrian N.\ Bishop}

\begin{document}

\maketitle

\begin{abstract}
Filtering and smoothing with a generalised representation of uncertainty is considered. Here, uncertainty is represented using a class of outer measures. It is shown how this representation of uncertainty can be propagated using outer-measure-type versions of Markov kernels and generalised Bayesian-like update equations. This leads to a system of generalised smoothing and filtering equations where integrals are replaced by supremums and probability density functions are replaced by positive functions with supremum equal to one. Interestingly, these equations retain most of the structure found in the classical Bayesian filtering framework. It is additionally shown that the Kalman filter recursion can be recovered from weaker assumptions on the available information on the corresponding hidden Markov model.
\end{abstract}

\begin{keywords}
  Outer measure, Information assimilation, Hidden Markov models
\end{keywords}

\begin{AMS}
  60A10,  	
  60J05,  	
  62L12  	
\end{AMS}

\section{Introduction}

The question of how to represent uncertainty is central when formulating any estimation, inference, or learning problem. This question has also long stirred debate among practitioners. Firstly, there was the frequentist versus Bayesian debate in early statistical estimation theory. Later, numerous attempts at ``generalising'' probabilistic concepts were derived and debated, such as fuzzy logic, imprecise probabilities, possibility theory, fuzzy random sets, and Dempster-Shafer theory \cite{Zadeh1965,Walley1991,Dempster1967,Shafer1976,Dubois1983,Yen1990,Friedman2001}. The proposed approach is closer in spirit to these latter methods, and assumes a specific structure that is general enough to cover most modelling needs and restrictive enough to enable the derivation of practical estimation algorithms. This approach is based on the fundamental measure-theoretic concept of an outer measure, which provides for a more relaxed manner of distributing probability mass. As explained in \cite{Fremlin2013}: €œThe idea of the outer measure of a set A is that it acts as an upper bound for the possible measure of A. This structure can in some sense capture standard probability theory, since a given outer measure can bound the probability mass of each measurable subset so finely that it collapses identically to a probability measure. By encompassing a broader spectrum of uncertainty, e.g. from pure randomness to totally non-random uncertainty, the presented estimation principle brings together the Bayesian and frequentist interpretation by simultaneously allowing for fixed randomness and evolving uncertainty based on the received information.

Practically, the proposed filtering/smoothing framework naturally accommodates a more relaxed model of the system dynamics, as well as the observed and prior information. This is achieved via the use of outer measures, and yields potentially more robust estimation algorithms that do not require all sources of uncertainty to be perfectly (and solely) described by strict probability distributions; e.g. Markov transition kernels in the case of the system dynamics. The language and nature of uncertainty may be important in certain applications. Closed-form recursive algorithms will be derived under this framework of outer measures for both filtering and smoothing, and using both forward and backward recursions (in time). An analogue of the classical Kalman filter recursion will also be derived under non-classical, and weaker, assumptions on the prior, dynamic, and observation models.

\section{Representation of uncertainty}

The objective in this section is to introduce a general representation of uncertainty based on \cite{Houssineau2015,Houssineau2016_dataAssimilation}, that relaxes the standard approach of defining probability distributions over the state space. The proposed approach will build on \cite{Houssineau2015,Houssineau2016_dataAssimilation} to enable filtering and smoothing recursions to be derived. The time is discrete and assumed to take integer values between $0$ and $T$ so that the set $\bbT$ of all time steps is defined as $\{0,\dots,T\}$. The state space at time $t\in\bbT$ is denoted $\boX_t$ and is assumed to be a subset of $\bbR^d$ for some $d > 0$. We first consider the problem of representing uncertainty on a single state space $\boE$, which might be $\boX_t$ at any time $t \in \bbT$, before tackling the case of the product space $\boX_{0:T} = \boX_0 \times \dots \times \boX_T$. The sets $\boE$ and $\boX_0,\dots,\boX_T$ are endowed with their respective Borel $\sigma$-algebra $\calB(\boE)$ and $\calB(\boX_0),\dots,\calB(\boX_T)$.

\subsection{On a single state space}

Instead of considering a probability distribution on $\boE$, we consider a probability measure\footnote{Measure-theoretic questions associated with the introduction of a measure on a set of functions are discussed in \cite{Houssineau2015}.} $P$ on the set $\boL(\boE)$ of measurable functions\footnote{The following convention is considered: the term \emph{mapping} or \emph{map} is used whatever the domain and co-domain while the term \emph{function} is reserved for real-valued maps, i.e.\ for maps whose co-domain is a subset of $\bbR$.} from $\boE$ to $\bbR^+$ with supremum equal to $1$, describing some knowledge about the system of interest in the following way: the underlying probability distribution $p$ on $\boE$, if any, is \emph{dominated} by $\oP$, i.e.\ it satisfies
\eqnl{eq:outerMeas}{
p(\tf) \leq \oP(\tf) = \int \ninf{\tf \cdot f} P(\d f),
}
for any $\tf$ in the set $\boL^{\infty}(\boE)$ of positive bounded measurable functions on $\boE$, where $\ninf{\cdot}$ is the supremum norm and where $\tf \cdot f$ denotes the point-wise product between $\tf$ and $f$, i.e.\ $(\tf \cdot f)(x) = \tf(x) f(x)$ for any $x \in \boE$. The reason for cautioning the existence of $p$, captured by the ``if any'' in the previous sentence, will be clarified later in this section. In particular, it follows from considering $\tf = \ind{B}$ in \eqref{eq:outerMeas} for some $B \in \calB(\boE)$ that
\eqnl{eq:outerMeas2}{
p(B) \leq \int \sup_{x \in B} f(x) P(\d f).
}
The set function $B \mapsto \oP(\ind{B})$ is a type of outer probability measure, that is a set function that gives value $0$ to the empty set, value $1$ to the whole space and that is monotone and countably sub-additive. The main difference with a probability measure being that the usual additivity assumption is replaced by sub-additivity. In the right hand side (r.h.s.) of \eqref{eq:outerMeas2}, the integral is additive by definition so it is the supremum that is responsible for the sub-additivity. For the sake of simplicity, we say that $\oP$ is an outer measure.

Defining measures on measurable subsets as in \eqref{eq:outerMeas2} or on measurable functions as in \eqref{eq:outerMeas} is equivalent \cite{Schwartz1973}, however it is not the case for outer measures because of their sub-additivity. Defining outer measures on measurable functions is then preferred since it is more general. Results with subsets can be recovered by considering an indicator function as in \eqref{eq:outerMeas2}.

The following examples aim at providing insight into the functions in $\boL(\boE)$ as well as into the difference between having one or several such functions in the support of $P$ (that is, informally, in the subset of $\boL(\boE)$ of functions to which $P$ gives positive probability).

\begin{example}
If $P = \delta_f$ for some function $f \in \boL(\boE)$ then it holds that
\eqns{
p(B) \leq \oP(\ind{B}) = \sup_{x \in B} f(x)
}
so that $f$ can be interpreted as a \emph{possibility function}\footnote{Possibility functions are called ``possibility distributions'' in the context of possibility theory \cite{Dubois2015} and the associated outer measures are called ``possibility measures''.} since it gives, through the supremum, an upper bound for the probability mass that $p$ can possibly give to $B$.
\end{example}

\begin{example}
Let $A$ and $A'$ be disjoint subsets of $\boE$ and consider the two following modelling choices:
\begin{enumerate}
\item $P = \delta_{\ind{A \cup A'}}$ which only implies that $p(B) = 0$ if $B$ is disjoint from both $A$ and $A'$, i.e.\ all the probability mass of $p$ is within $A \cup A'$ but it could as well be all in $A$ or all in $A'$.
\item $P = a \delta_{\ind{A}} + (1-a) \delta_{\ind{A'}}$ for some $a \in (0,1)$ which implies $p(A) = a$ and $p(A') = (1-a)$, i.e.\ the probability can be distributed in any way within $A$ and $A'$ as long as the total mass in $A$ is equal to $a$.
\end{enumerate}
\end{example}

\begin{example}
If the space $\boE$ is discrete, say equal to $\bbN$, then $p$ can be characterised by its probability mass function $m$ via
\eqns{
m(n) \leq \int f(n) P(\d f)
}
for any $n \in \bbN$. In particular, if $P$ is of the form
\eqns{
P = \sum_{i \in I} w_i \delta_{f_i}
}
for indexed family $\{(w_i,f_i)\}_{i \in I}$ of positively-weighted functions in $\boL(\boE)$ then
\eqns{
m(n) \leq \sum_{i \in I} w_i f_i(n).
}
We can notice that if $I = \bbN$ and if $f_i = \ind{n}$ then it holds that $m(n) \leq w_n$ for any $n \in \bbN$. However, since both $m(n)$ and $w_n$ sum to $1$ by definition, it follows that $m(n) = w_n$; in other words $P$ is equivalent to the probability measure $p$. The same approach can be used for uncountable spaces at the cost of measure theoretic notations.
\end{example}

Many distributions that are commonly used in statistics have an analogous possibility function. The interpretation of the two is however different since one fully characterises randomness while the other only suggests a given shape for the uncertainty. Notice that, as opposed to probability density functions on continuous spaces, possibility functions do not require a reference measure to be defined in order to be written as a function.

\begin{definition}
\label{def:GaussianPossibility}
A function $f$ in $\boL(\bbR^d)$ is said to be a \emph{Gaussian possibility function} if it takes the form
\eqns{
f(x) = \bar\calN(x; m,\bsP) \defeq \exp\Big(-\frac{1}{2} (x-m)^T \bsP^{-1}(x-m)\Big),
}
for some $m \in \bbR^d$ and for some $d\times d$ positive-definite matrix $\bsP$ with real coefficients.
\end{definition}

We refer to the parameters $m$ and $\bsP$ of $\bar\calN(\cdot; m,\bsP)$ as the ``mean'' and the ``spread'' of the Gaussian possibility function. In this context, referring to $m$ as the mean may be viewed as a slight, but useful, abuse of terminology. 

The distribution $P$ on $\boL(\boE)$ can encode information in a Bayesian and/or a frequentist way. If the embedded information relates to a non-random phenomenon, either as a realisation of a random variable or as a fully non-random parameter, then there is no underlying probability measure and $P$ describes the uncertainty in a Bayesian sense. However, if the embedded information relates to a random variable (an actual one, not a realisation of it), then the underlying probability measure exists and is unique, and $P$ describes (partially-unknown) randomness as in the frequentist interpretation.

In this context, it is better to understand random variables as solely describing randomness with another concept needed to describe uncertainty in general. The concept of \emph{uncertain variable} is therefore introduced in order to describe a variable about which $P$ gives information. Two sample spaces are required, the first, denoted $(\Omega_{\r},\calF)$, describes randomness and the second, denoted $\Omega_{\u}$, describes non-random uncertainty. Only the former is endowed with a $\sigma$-algebra $\calF$ since non-random events (subsets of $\Omega_{\u}$) do not need to be assigned a probability mass, it is just unknown whether or not they have happened. The sample space $(\Omega_{\r},\calF)$ is endowed with a probability measure $\bbP(\cdot\given \omega_{\u})$ conditioned on the state $\omega_{\u} \in \Omega_{\u}$ of all non-random phenomenon. The sample space $\Omega_{\u}$ can be seen as a space $\Theta$ describing (possibly unknown) parameters, so that the probability measure $\bbP(\cdot\given \theta)$ with $\theta \in \Theta$ can simply be seen as a parametrised distribution. This separation of random and non-random phenomena imply that degenerate random variables are not considered as random variables but as parameters. Uncertain variables can now be defined straightforwardly by using the sample spaces $(\Omega_{\r},\calF,\bbP)$ and $\Omega_{\u}$.

\begin{definition}
\label{def:uncertainVariable}
An uncertain variable $X$ on a measurable space $(\boE,\calE)$ is a mapping between the product sample space $\Omega = \Omega_{\r}\times \Omega_{\u}$ and $\boE$ such that the mapping $\omega_{\r} \mapsto X(\omega_{\r},\omega_{\u})$ is measurable for every $\omega_{\u} \in \Omega_{\u}$.
\end{definition}

An uncertain variable $X$ reduces to a random variable when the mapping $\omega_{\u} \mapsto X(\omega_{\r},\omega_{\u})$ is constant. Alternatively, if the mapping $\omega_{\r} \mapsto X(\omega_{\r},\omega_{\u})$ is constant then $X$ is a non-random uncertain variable and the measurability condition in \cref{def:uncertainVariable} is always satisfied.

If an uncertain variable is not a random variable, then there is no underlying probability measure on $\boE$ that would be dominated by the outer measure $\oP$, instead, the latter describes the uncertain variable directly, e.g.\ the scalar $\oP(\ind{B})$ gauges how likely the event $X \in B$ is for any $B \in \calB(\boE)$. However, it is sometimes useful to consider a (non-unique) probability measure $p$ dominated by $\oP$ in order to understand how operations on the state space $\boE$ affect the outer measure $\oP$, as in \cref{ssec:jointOuterMeas} below.

There is a natural transfer from randomness to non-random uncertainty as random phenomena take place and induce uncertainty about the corresponding realisations. For instance, if a coin is being flipped then it is usual to consider the output as random, however there is no more randomness once the coin has landed, and only uncertainty is left (at least until the outcome is observed).

Note that an outer measure $\oP$ on $\boE$ describing what is known about an uncertain variable $X : \Omega \to \boE$ can be pulled back \cite{Houssineau2016_dataAssimilation} onto $\Omega$ (see also \cref{sec:operations} below for the concept of pullback measure). The resulting outer measure $\frakP$ can be seen as an extrinsic description of the uncertainty whereas $\bbP$ is an intrinsic characterisation of the randomness, the former changes when more information is acquired while the latter never changes. The whole sample space could then be seen as $(\Omega,\frakP)$, where $\frakP$ is a ``generalised'' probability which does not satisfy Kolmogorov's third axiom ($\sigma$-additivity). However, we do not emphasize this interpretation.

Modelling single-variate/unconditional uncertainty as in this section can be sufficient for many applications, e.g.\ in expert systems \cite{Zadeh1983} or to derive data-association formulas \cite{Houssineau2015} for multi-object representations~\cite{Houssineau2016_population}. However, conditioning is key to express smoothing and filtering equations so that the proposed modelling has to be extended further.

\subsection{On a joint state space}
\label{ssec:jointOuterMeas}

The focus is now on the product space $\boX_{0:T}$ and most of the results will be stated accordingly. Yet, when introducing notations and concepts that apply more broadly, the set $\boE$ will be reused together with another set $\boF$, also endowed with its Borel $\sigma$-algebra. Equation~\eqref{eq:outerMeas} can be extended to the product space $\boX_{0:T}$ as
\eqnl{eq:jointOuterMeas}{
p_{0:T}(\bstf) \leq \oP_{0:T}(\bstf) =  \int \ninf{\bstf \cdot f} P_{0:T}(\d f),
}
for any function $\bstf$ in $\boL^{\infty}(\boX_{0:T})$, where $p_{0:T}$ and $P_{0:T}$ are probability measures on $\boX_{0:T}$ and $\boL(\boX_{0:T})$ respectively. The possibility $f$ can be thought of as a \emph{joint possibility} since it jointly applies throughout the different state spaces $\boX_0$ to $\boX_T$. However, the outer measure $\oP_{0:T}$ will prove to be insufficient in practice. For instance, information at $t=1$ might be conditional on the state at $t=0$, information at $t=2$ might be conditional at both the states at $t=0$ and $t=1$, etc. In this case, we have to introduce a more general type of outer measure. For this purpose, let $P_t( \cdot \given x_{0:t-1})$ be the conditional distribution on $\boL(\boX_t)$ at time $t \in \bbT$ given the states\footnote{The term $x_{0:t-1}$ stands for the sequence $(x_0,\dots,x_{t-1})$, which is the empty sequence when $t=0$.} $x_{0:t-1}$ at all previous times and let $p_t(\cdot \given x_{0:t-1})$ be a probability distribution on $\boX_t$ verifying
\eqns{
p_t(\tf \given x_{0:t-1}) \leq \int \ninf{\tf \cdot f} P_t(\d f \given x_{0:t-1}),
}
for any $\tf \in \boL^{\infty}(\boX_t)$. Now, let $\oP_t(\bstf)$ be a \emph{conditional outer measure} defined for any $\bstf \in \boL^{\infty}(\boX_{0:t})$ as
\eqns{
\oP_t(\bstf)(x_{0:t-1}) = \int \sup_{x_t \in \boX_t} (f(x_t) \bstf(x_{0:t})) P_t(\d f \given x_{0:t-1}),
}
for all $x_{0:t-1} \in \boX_{0:t-1}$, then we have the result of the following theorem about the joint probability $p_{0:T}$, in which $\oP\oP'(\bstf)$ denotes $\oP(\oP'(\bstf))$ for any outer measures $\oP$ on $\boE$, any conditional outer measure $\oP'(\cdot)(x)$ on $\boF$ defined for all $x \in \boE$ and for any $\bstf \in \boL^{\infty}(\boE\times \boF)$.

\begin{theorem}
\label{thm:condOuterMeasure}
The outer measure $\oP_{0\up T}$ induced by the family $\{P_t(\cdot \given x_{0:t-1})\}_{t\in\bbT}$ of probability distributions is characterised by
\eqnl{eq:condOuterMeas}{
\oP_{0\up T}(\bstf) = \oP_0 \dots \oP_T(\bstf),
}
for any $\bstf \in \boL^{\infty}(\boX_{0:T})$.
\end{theorem}

Before proving \cref{thm:condOuterMeasure}, it is useful to note that the way in which the conditioning is defined in the family $\{P_t(\cdot \given x_{0:t-1})\}_{t\in\bbT}$ matters for the corresponding outer measure. If we consider another family of distributions of the form $\{P'_t(\cdot \given x_{t+1:T})\}_{t\in\bbT}$, then the associated outer measure would be
\eqns{
\oP_{T\down 0}(\bstf) = \oP'_T \dots \oP'_0(\bstf),
}
which differs from $\oP_{0\up T}(\bstf)$ in general. This is one of the properties of probability measures that does not transfer to outer probability measures. In the context of filtering, we will be mostly interested in conditioning with respect to the past so that $\oP_{0\up T}$ will be used predominantly.

\begin{proof}
The result is obvious for $T = 0$. Let $p_{0:T}$ denote a probability distribution on $\boX_{0:T}$ that is induced by a family $\{p_t(\cdot \given x_{0:t-1})\}_{t\in\bbT}$ of probability distribution dominated by $\{\oP_t(\cdot \given x_{0:t-1})\}_{t\in\bbT}$ and assume that the results holds for the set $\{0, \dots, T-1\}$, then the information can be summed up through the two following inequalities:
\eqnsa{
p_{0:T-1}(\bstf) & \leq \oP_{0\up T-1}(\bstf), \\
p_T(\tf \given x_{0:T-1}) & \leq \int \ninf{\tf \cdot f_T} P_T(\d f_T \given x_{0:T-1}),
}
which hold for any $\bstf \in \boL^{\infty}(\boX_{0:T-1})$, any $\tf \in \boL^{\infty}(\boX_T)$ and any $x_{0:T-1} \in \boX_{0:T-1}$. It follows that
\eqnsa{
p_{0:T}(\bstf') & = \int \bstf'(x_{0:T}) p_T(\d x_T \given x_{0:T-1}) p_{0:T-1}(\d x_{0:T-1}) \\
& \leq \oP_{0\up T-1}\big(p_T\big( \bstf' \given \cdot \big)\big)
}
for any $\bstf' \in \boL^{\infty}(\boX_{0:T})$, so that
\eqnsa{
p_{0:T}(\bstf') & \leq \oP_{0\up T-1}\bigg( \int \sup_{x_T \in \boX_T} \big( \bstf'(\cdot,x_T) f(x_T) \big) P_T(\d f \given \cdot) \bigg) \\
&= \oP_{0\up T-1}(\oP_T(\bstf')) = \oP_{0\up T}(\bstf').
}
This concludes the proof of the theorem.
\end{proof}

\begin{remark}
\label{rem:independentOuterMeas}
If for any $t \in \bbT$, the distribution $P_t(\cdot \given x_{0:t-1})$ does not actually depend on $x_{0:t-1}$ and if a subset $B$ of the form $B = B_0 \times \dots \times B_T$ is considered then \eqref{eq:condOuterMeas} collapses to
\eqnl{eq:rem:independentOuterMeas}{
\oP_{0\up T}(\ind{B}) = \oP_0(\ind{B_0}) \dots \oP_T(\ind{B_T}).
}
In this case, for any separable function $\bstf(\bsx) = \bstf_1(\bsx_1)\dots\bstf_T(\bsx_T)$ in $\boL^{\infty}(\boX_t)$, it holds that $\oP_{0\up T}(\bstf) = \oP_{T\down 0}(\bstf)$ since the individual terms in \eqref{eq:condOuterMeas} can now be commuted.
\end{remark}

\begin{definition}
Let $X$ and $X'$ be two uncertain variables on the respective spaces $\boE$ and $\boF$ and let $\oP$ be an outer measure describing information about the joint $(X,X')$ on the product space $\boE \times \boF$. If there exist outer measures $\oP_X$ and $\oP_{X'}$ such that for every separable function $\bstf = \tf \times \tf'$ in $\boL^{\infty}(\boE \times \boF)$ it holds that
\eqns{
\oP(\bstf) = \oP_X(\tf) \oP_{X'}(\tf'),
}
then $X$ and $X'$ are said to be \emph{weakly independent}.
\end{definition}

If two uncertain variables are at least partially random then their weak independence is unrelated to their statistical independence. Weak independence only means that the relation between the two uncertain variables is unknown. It is therefore meaningful to introduce a third concept, \emph{strong independence}, to describe the case where two possibly-random uncertain variables are (statistically) independent as well as weakly independent. To sum up, in terms of independence
\eqns{
\text{strongly} \iff \text{statistically and weakly},
}
however neither does statistical independence imply weak independence nor the other way around.

\begin{example}
If for any time $t \in \bbT$, it holds that $P_t(\cdot \given x_{0:t-1}) = \delta_{f_t(\cdot \given x_{0:t-1})}$ where $f_t$ is a \emph{conditional possibility}, i.e.\ it is such that $f_t(\cdot \given x_{0:t-1}) \in \boL(\boX_t)$ for any $x_{0:t-1} \in \boX_{0:t-1}$, then
\eqns{
\oP_{0:T}(\bstf) = \ninf{ \bstf \cdot f_{0:T}(\bsx)}
}
for any $\bstf \in \boL^{\infty}(\boX_{0:T})$, where $f_{0:T}$ is a joint possibility in $\boL(\boX_{0:T})$ defined as
\eqns{
f_{0:T} : x_{0:T} \mapsto f_0(x_0) f_1(x_1 \given x_0) \dots f_T(x_T \given x_{0:T-1}).
}
\end{example}

\subsection{Hidden Markov models and outer measures}

Hidden Markov models are often considered when dealing with estimation for dynamical systems \cite{Cappe2005}. It is therefore of interest to generalise the concept of Markov chain to outer measures describing collections of uncertain variables. We propose the following approach: the uncertainty about a collection $\{X_t\}_{t \in \bbN}$ of uncertain variables has the Markov property if for any $t \in \bbN$ it holds that
\eqns{
\oP_t (\tf)(x_{0:t-1}) = \oP_t(\tf)(x_{t-1})
}
for all $x_{0:t-1} \in \boX_{0:t-1}$ and for all $\tf \in \boX_t$. Note that in this case, the property is more about the available knowledge than about the uncertain variables themselves. For instance, at the final time $T$, we might be given information that the physical system has never been twice in the same state, in which case the uncertainty would cease to have the Markov property.

\subsubsection*{Principle}

It is assumed that the system of interest can be characterised by a collection $\{X_t\}_{t\in\bbT}$ of uncertain states and its observation is described by a collection $\{Y_t\}_{t\in\bbT}$ of observations on the space $\boY_t$. This can be summed up by
\eqnl{eq:estimation}{
\begin{cases}
X_t = F_t(X_{t-1},V_t) \\
Y_t = O_t(X_t)
\end{cases}
}
where $F_t$ and $O_t$ are respectively the state transition and the observation function at time $t$ and where $\{V_t\}_{t\in\bbT}$ is a sequence of strongly independent uncertain variables. In some cases, the uncertain variable $V_t$ will be equivalent to a random variable, e.g.\ when describing the motion of particle in turbulent water, however, in many other cases, it will represent an unknown but non-random change in the model, e.g.\ a plane manoeuvring or a pedestrian changing direction.

The mechanism behind the acquisition of information through the observation process is different. We assume that the observation is non-random\footnote{This assumption is not crucial, it is only used to simplify the following statements. Alternatively, an additional Markov kernel $S_t(x_t,\cdot)$ on $\boL(\boY_t)$ can be defined for any $x_t \in \boX_t$ in order to model the knowledge about a perturbed observation function~$O_t(\cdot,W_t)$ where $\{W_t\}_{t \in \bbT}$ is a sequences of strongly independent uncertain variables that are independent of $\{V_t\}_{t\in\bbT}$.} but the usual assumption that $Y_t$ is received directly is not considered. Instead, it is assumed that what is received is information about $Y_t$ rather than $Y_t$ itself. Information about $Y_t$ is given under the form of an outer measure on $\boY_t$ (this representation will be formalised subsequently). For example, the observation $Y_t$ may be known to be in some subset of the observation space (e.g.\ $Y_t \in A$ where $A$ is one or several pixels of a camera) or information about $Y_t$ may be given more indirectly under the form of a natural language statement. Numerous other modelling examples can be considered.

Since we have assumed that all the information that will be available about the collection $\{X_t,Y_t\}_{t \in \bbT}$ will have the Markov property, the overall choice of model is still referred to as a hidden Markov model.

\subsubsection*{Formalisation}

The uncertainty about the collection $\{X_t\}_{t\in\bbT}$ has the Markov property by construction since $X_t$ only depends on $X_{t-1}$ and since the uncertainty induced by $V_t$ is independent of the one at previous times. The transition can therefore be encoded into a Markov kernel $Q_t(x_{t-1},\cdot)$ on $\boL(\boX_t)$, which contains information on $\boX_t$ conditional on the state in $x_{t-1} \in \boX_{t-1}$ but not directly on $\boX_{t-1}$ or any previous state space.

The information about the observation $Y_t$ at time $t$ is assumed to be weakly independent of the information at other times and is given under the form of a probability measure $R_t$ on $\boL(\boY_t)$, whose projection onto $\boL(\boX_t)$ by $O_t$ is the pullback measure $O_t^* R_t$ (it can also be assumed for simplicity that $\boY_t = \boX_t$ and $O_t$ is the identity).

The initial state $X_0$ is an uncertain variable described by the outer measure $\oP_0$ induced by the distribution $P_0$ on $\boL(\boX_0)$.

\section{Translating operations of probability theory to outer measures}
\label{sec:operations}

The equivalent of the standard operations of probability theory have to be derived for the considered class of outer measures in order to generalise the usual filtering and smoothing algorithms. Other useful operations that do not exist under the usual framework are also introduced. We start with the equivalent of the push-forward $\xi_* p = p(\xi^{-1}(\cdot))$ of a probability measure $p$ on $\boE$ by a measurable mapping $\xi$ from $\boE$ to another set $\boF$. For a given outer measure $\oP$, the objective is to characterise the outer measure $\oP'$ verifying
\eqnl{eq:condPushforward}{
\oP'(\ind{B}) = \oP(\ind{\xi^{-1}[B]}),
}
for any appropriate subset $B$ of $\boF$, where the use of square brackets in the expression $\xi^{-1}[B]$ of the inverse image of the subset $B$ by $\xi$ emphasizes that the result is a set. The solution is given in the next proposition.

\begin{proposition}[From {\cite[Proposition 3]{Houssineau2016_dataAssimilation}}]
Let $P$ be a distribution on $\boL(\boE)$ and let $\xi$ be a measurable mapping from $(\boE,\calB(\boE))$ to $(\boF,\calB(\boF))$, then the probability distribution $P'$ which implies that \eqref{eq:condPushforward} holds is equal to the push-forward $\vec{\xi}_* P$ where $\vec{\xi}$ is a mapping from $\boL(\boE)$ to $\boL(\boF)$ defined as
\eqns{
\vec{\xi}(f) : x \mapsto \sup_{\xi^{-1}[x]} f
}
for any $f \in \boL(\boE)$.
\end{proposition}

The term $\xi^{-1}[x]$, which is shorthand for $\xi^{-1}[\{x\}]$, is a set which is a singleton if $\xi$ is bijective. As a consequence, the simplest examples of this equivalent of push-forward are found when $\xi$ is bijective as explained in the following remark. 

\begin{remark}
If $\xi$ is bijective then the expression of $\vec{\xi}$ simplifies to $\vec{\xi}(f) = f \circ \xi^{-1}$ for any $f \in \boL(\boE)$. Therefore, $P'$ gives mass $P(\d f)$ to the function $f \circ \xi^{-1}$, so that, for instance, if $\xi$ is the transformation from Cartesian coordinates in $\boE = \bbR^2\setminus \{(0,0)\}$ to polar coordinates systems in $\boF = [0,2\pi) \times (0,\infty)$ and if $f$ is the indicator of the disk $\{(x,y) \in \boE \st \sqrt{x^2+y^2} \leq 1 \}$ then $f \circ \xi^{-1}$ is simply the indicator of the set $\{ (a,r) \in \boF \st r \leq 1 \}$.
\end{remark}

Very often, the interest lies in non-bijective mappings with the simple case of a projection being studied in the following example.

\begin{example}
\label{ex:projection}
If $\boE = \boX_{t-1} \times \boX_t$ and $\boF = \boX_t$ for some time $t \in \{1,\dots,T\}$ and if $\xi$ is the canonical projection $(x_{t-1},x_t) \mapsto x_t$, then
\eqns{
\vec{\xi}(f) : x_t \mapsto \sup_{x \in \boX_{t-1}} f(x,x_t)
}
for any $f \in \boL(\boX_{t-1} \times \boX_t)$. This operation can be seen as marginalisation for possibility functions. As will be practically verified later, operations for possibility functions are the same as for probability density functions except that integrals are replaced by supremums. The consequence for the outer measures $\oP'$ is that, for any $\tf \in \boL^{\infty}(\boX_t)$, it holds that
\eqns{
\oP'(\tf) = \int \sup_{(x',x) \in \boX_{t-1} \times \boX_t} \big( \tf(x) f(x',x) \big) P(\d f),
}
which can be written as $\oP'(\tf) = \oP(\tf)$ by seeing $\tf$ as the function on $\boX_{t-1} \times \boX_t$ such that $\tf : (x',x) \mapsto \tf(x)$ in the r.h.s.\ of the equality. This abuse of notations will be used when there is no possible confusion. If $\oP$ is the single-possibility outer measure verifying $P = \delta_{f_{t-1,t}}$ for some $f_{t-1,t} \in \boL(\boX_{t-1} \times \boX_t)$ then, using obvious notational choices, it can be written that
\eqns{
f_t(x_t) = \sup_{x_{t-1} \in \boX_{t-1}} f_{t-1,t}(x_{t-1},x_t),
}
where $f_t$ is the possibility function such that $P' = \delta_{f_t}$. This result motivates the choice of performing operations directly on possibility functions in the case of single-possibility outer measures.
\end{example}

In standard probability theory, the inverse of the push-forward operation by $\xi$ applied to a probability measure $p$ is ill-defined since there might be several \emph{pullback} measures $p'$ verifying the identity $\xi_* p' = p$. However, all these probability measures are dominated by a given outer measure, so that the operation is meaningful for the latter. For a given outer measure $\oP$, the objective is to characterise the outer measure $\oP'$ verifying
\eqnl{eq:condPullback}{
\oP'(\ind{B}) = \oP(\ind{\xi[B]}),
}
for any appropriate subset $B$. Note that although $\xi[B]$ might not be measurable even if $B$ is, outer measures are not limited to measurable subsets. The solution is given in the next proposition.

\begin{proposition}[From {\cite[Proposition 4]{Houssineau2016_dataAssimilation}}]
Let $P$ be a distribution on $\boL(\boF)$ and let $\xi$ be a measurable mapping from $(\boE,\calB(\boE))$ to $(\boF,\calB(\boF))$, then the probability distribution $P'$ which implies that \eqref{eq:condPullback} holds is equal to the push-forward $\cev{\xi}_* P$ where $\cev{\xi}$ is the mapping
\eqnsa{
\cev{\xi} : \boL(\boF) & \to \boL(\boE) \\
f & \mapsto f \circ \xi.
}
\end{proposition}

For the sake of compactness, the push-forward $\cev{\xi}_* P$ can be denoted $\xi^* P$ since there is no possible confusion with existing notations. The distribution $\xi^* P$ is called the \emph{pullback} of $P$ by $\xi$. The case of a bijective $\xi$ is not so interesting here since the pullback is the same of the push-forward by the inverse in this case. The projection studied in \cref{ex:projection} is however of central interest since there is no equivalent for probability measures in this case. The following example is in the continuation of \cref{ex:projection}.

\begin{example}
\label{ex:inverseProjection}
If $\boE = \boX_{t-1} \times \boX_t$ and $\boF = \boX_t$ for some time $t \in \{1,\dots,T\}$ and if $\xi$ is the canonical projection $(x_{t-1},x_t) \mapsto x_t$, then
\eqns{
\cev{\xi}(f) : (x_{t-1},x_t) \mapsto f(\xi(x_{t-1},x_t)) = f(x_t)
}
for any $f \in \boL(\boX_t)$. This operation is indeed the inverse of marginalisation where no knowledge on the added state space is assumed. The consequence for the outer measures $\oP'$ is that
\eqns{
\oP'(\bstf) = \int \sup_{(x',x) \in \boX_{t-1} \times \boX_t} \big( \bstf(x',x) f(x) \big) P(\d f),
}
for any $\bstf \in \boL^{\infty}(\boX_{t-1} \times \boX_t)$. If $\oP$ is the single-possibility outer measure verifying $P = \delta_{f_t}$ for some $f_t \in \boL(\boX_t)$ then, using obvious notational choices, it can be written that
\eqns{
f_{t-1,t}(x_{t-1},x_t) = f_t(x_t),
}
where $f_{t-1,t}$ is the possibility function such that $P' = \delta_{f_{t-1,t}}$. It follows that $\sup f_{t-1,t}(x_{t-1},\cdot) = 1$ for any $x_{t-1} \in \boX_{t-1}$, which means that nothing is known on $\boX_{t-1}$ as expected.
\end{example}

Continuing in the spirit of examples~\ref{ex:projection} and \ref{ex:inverseProjection}, it can be verified that if $f$ is a possibility function on $\boX_{t-1} \times \boX_t$ then there exists a function $f_{t-1} \in \boL(\boX_{t-1})$ and a function $f_{t|t-1}(\cdot \given x')$ on $\boL(\boX_t)$ for every $x' \in \boX_{t-1}$ such that
\eqns{
f(x_{t-1},x_t) = f_{t|t-1}(x_t \given x_{t-1}) f_{t-1}(x_{t-1}),
}
for any $(x_{t-1},x_t) \in \boX_{t-1} \times \boX_t$, so that
\eqns{
f_{t|t-1}(x_t \given x_{t-1}) = \dfrac{f(x_{t-1},x_t)}{f_{t-1}(x_{t-1})} = \dfrac{f(x_{t-1},x_t)}{\displaystyle \sup_{x' \in \boX_{t-1}} f(x',x_t)}
}
Once again, the usual operations of probability theory can be seen to hold for possibility functions with integrals replaces by supremums. The analogue of Bayes' theorem on the state space for the considered class of outer measures, however, will be seen to take a different form in the next section.

\section{Information assimilation}

In order to describe the result of the combination of two strongly independent pieces of information, an additional notation has to be introduced: if $f \in \boL^{\infty}(\boE)$ then $f^{\scl} = f/\ninf{f} \in \boL(\boE)$ is the rescaled version of $f$ which has supremum $1$.

\begin{theorem}[From {\cite[Theorem 1]{Houssineau2016_dataAssimilation}}]
\label{thm:fusion}
Let $P$ and $P'$ be two probability measures on $\boL(\boE)$ describing respectively the uncertain variables $X$ and $X'$. If $X$ and $X'$ are strongly independent, then the posterior distribution $P \star P'$ based on $P$ and $P'$ can be expressed as
\eqnl{eq:fusion}{
P \star P'(F) \defeq \dfrac{ \int \ind{F}( (f\cdot f')^{\scl} ) \ninf{f \cdot f'} P(\d f) P'(\d f') }{ \int \ninf{f\cdot f'} P(\d f) P'(\d f') }
}
for any measurable subset $F$ of $\boL(\boE)$ as long as $P$ and $P'$ are \emph{compatible}, i.e.\ as long as the denominator is strictly positive.
\end{theorem}

The strong independence considered in \cref{thm:fusion} is analogous to the statistical independence assumed in the standard Bayes' theorem. The denominator of \eqref{eq:fusion} is a scalar in the interval $(0,1]$ and quantifies how likely it is that $P$ and $P'$ represent the same system. The rescaling $\cdot^{\scl}$ ensures that $P \star P'$ is a probability measure supported by possibility functions rather than an arbitrary measure supported by arbitrary functions of the form $f \cdot f'$ for some $f,f' \in \boL(\boE)$. Rescaling is not necessary if the outer measure $\overline{P \star P'}$ induced by $P \star P'$ is considered instead, since it can simply be written that
\eqnl{eq:fusionOuterMeas}{
\overline{P \star P'}(\tf) \propto \int \ninf{\tf \cdot f \cdot f'} P(\d f) P'(\d f')
}
for any $\tf \in \boL^{\infty}(\boE)$. Several special cases of the use of the operation $\star$ are given in \cite{Houssineau2016_dataAssimilation}.

\begin{example}
If $X$ and $X'$ are uncertain variables that (at least partially) characterise the same random phenomenon, i.e.\ they have some statistical dependence, then the associated outer measures $P$ and $P'$ cannot be fused together using \cref{thm:fusion}. For instance, if two observers study a biased coin and independently determine that the probability of heads is $3/4$ then it is erroneous to combine these information and conclude that the probability of heads must be $(3/4 \times 3/4) / (1/4 \times 1/4 + 3/4 \times 3/4) = 9/10$. However, if the coin is tossed and two observers witness the experiment but are unsure of the outcome, e.g.\ they are both $75\%$ sure that the result was tails, then it is possible to combine these independent pieces of information and claim that the outcome was tails with a likelihood of $9/10$. This result also holds if one observer has studied the coin and the other has independently witnessed the experiment.
\end{example}

Both $P$ and $P'$ can be seen as priors and the probability measure $P \star P'$ can be seen as a Bayesian posterior given that $P$ and $P'$ represent the same system. This can be highlighted by assuming that the system of interest is fully characterised by its state in $\boE$, so that the event ``$P$ and $P'$ represent the same system'' corresponds to the diagonal $\Delta$ of $\boE \times \boE$. In this case, we can define a joint probability measure $\breve{P} = P \times P'$ and a likelihood $\ell(\Delta \given f,f') = \ninf{f \cdot f'}$ giving the compatibility between $f$ and $f'$, e.g.\ $\ell(\Delta \given \ind{A},\ind{A'}) = 0$ if $A$ and $A'$ are disjoint subsets. With these notations, we can compute the posterior distribution
\eqns{
\hat{P}(\hat{F} \given \Delta) \defeq \dfrac{ \int_{\hat{F}} \ell(\Delta \given \bsf) \breve{P}(\d \bsf)}{\int \ell(\Delta \given \bsf) \breve{P}(\d \bsf)}
}
for any measurable subset $\hat{F}$ of $\boL(\boE) \times \boL(\boE)$. However, since we are only interested in the value of the function $\bsf = (f,f')$ on the diagonal $\Delta$, i.e.\ the values of the function $\hat{f}(x) = f(x)f'(x)$, we introduce a kernel $K((f,f'),F) = \delta_{(f\cdot f')^{\scl}}(F)$ which projects compatible possibilities to a single posterior possibility, and the distribution $P \star P'$ on $\boL(\boE)$ is found to be equal to the projection of $\hat{P}(\cdot \given \Delta)$ in the following way: 
\eqns{
P \star P'(F) = \int K(\bsf,F) \hat{P}(\d \bsf \given \Delta)
}
for any measurable subset $F$ of $\boL(\boE)$. The presence of the kernel $K$ is not usual, but it is just a projection, and the usual ingredients of Bayes theorem such as the prior $\breve{P}$ and the likelihood $\ell(\Delta \given \cdot)$ can be identified.

If $\oP_{0\up T}$ and $\oP'_{0\up T}$ are two joint outer measures induced by $\{P_t(\cdot \given x_{0:t-1})\}_{t\in\bbT}$ and $\{P'_t(\cdot \given x_{0:t-1})\}_{t\in\bbT}$ then the operation $\star$ can be applied component-wise and gives the posterior joint outer measure $\oP^{\star}_{0\up T}$ characterised by
\eqns{
\oP^{\star}_{0\up T}(\bstf) = \oP^{\star}_0 \dots \oP^{\star}_T(\bstf),
}
for any $\bstf \in \boL^{\infty}(\boX_{0:T})$, with $\oP^{\star}_t(\cdot)(x_{0:t-1})$ the outer measure on $\boX_t$ induced by $P_t(\cdot \given x_{0:t-1}) \star P'_t(\cdot \given x_{0:t-1})$ for any $t \in \bbT$. Note that in general, $\oP_{0\up T}$ can also be combined with an outer measure of the form $\oP'_{0:T}$ but not with one of the form $\oP'_{T\down 0}$.

\section{Smoothing}
\label{sec:smoothing}

The objective in this section is to derive an expression of the posterior outer measure on the joint space $\boX_{0:T}$ induced by the combination of all the information available up to time $T$. The Markov property is not sufficient to simplify the predicted outer measure $\oP_{0\up T}$ on $\boX_{0:T}$ which takes the form
\eqns{
\oP_{0\up T}(\bstf) = \oP_0 \oQ_1 \dots \oQ_T(\bstf),
}
for any $\bstf \in \boL^{\infty}(\boX_{0:T})$. The observed information across time can be expressed as another outer measure $\oR_{0:T}$ on $\boY_{0:T}$ characterised by
\eqns{
\oR_{0:T}(\bstf) = \oR_0(\tf_0) \dots \oR_T(\tf_T),
}
for any separable function $\bstf(\bsy) = \tf_0(\bsy_0) \dots \tf_T(\bsy_T)$ in $\boL^{\infty}(\boY_{0:T})$. This can also be expressed through a single probability distribution $R_{0:T}$ on $\boL(\boY_{0:T})$ defined as the product $R_{0:T} = R_0 \times \dots \times R_T$. The smoothed outer measure $\oP_{0\up T|T}$ is the posterior outer measure based on $\oP_{0\up T}$ and $\oR_{0:T}$, that is
\eqns{
\oP_{0\up T|T}(\bstf) = \oP_{0|0} \oQ_{1|1} \dots \oQ_{T|T}(\bstf),
}
for any $\bstf \in \boL^{\infty}(\boX_{0:T})$, where $\oP_{0|0}$ is the outer measure induced by $P_0 \star (O_0^*R_0)$ and where $\oQ_{t|t}(\cdot)(x_{t-1})$ is the conditional outer measure induced by $Q_t(x_{t-1},\cdot) \star (O_t^* R_t)$ for any $x_{t-1} \in \boX_{t-1}$ and for all $t \in \{1,\dots,T\}$.

One way of simplifying the form of $\oP_{0\up T|T}$ is to make the composition of ``$\int\sup (\cdot) P(\d f)$'' collapse by retaining a single term in each integral as in the following theorem. The other natural way is to cancel out the supremums, but this requires $P_0$ to be equivalent to a probability measure on $\boX_0$ and all the Markov kernels $Q_t$ to be equivalent to Markov kernels on $\boX_t$, which leads back to a formulation that resemble the classical Bayesian formulation (except for the observed information).

\begin{theorem}
\label{thm:smoothing}
If for any $t \in \{1,\dots,T\}$ there exists a function $g_t(x,\cdot) \in \boL(\boX_t)$ such that $Q_t(x,\cdot) = \delta_{g_t(x,\cdot)}$ for any $x \in \boX_{t-1}$, then the smoothed outer measure $\oP_{0:T|T}$ is characterised by
\eqns{
\oP_{0:T|T}(\bstf) \propto \int \ninf{ \bstf \cdot u_{0:T|T} }  P_0(\d f_0) R_0(\d h_0) \dots R_T(\d h_T),
}
for any $\bstf \in \boL^{\infty}(\boX_{0:T})$, where $u_{0:T|T} \in \boL^{\infty}(\boX_{0:T})$ depends implicitly on $f_0$ and $h_0,\dots,h_T$ and is characterised by
\eqns{
u_{0:T|T}(\bsx) = f_0(\bsx_0) \prod_{t=1}^T g_t(\bsx_{t-1},\bsx_t) \prod_{t=0}^T h_t(O_t(\bsx_t)).
}
for every $\bsx \in \boX_{0:T}$.
\end{theorem}

The statement of \cref{thm:smoothing} involves the function $u_{0:T|T}$ that is not a possibility function in general as in \eqref{eq:fusionOuterMeas}. This is only for the sake of compactness as the rescaled version of $u_{0:T|T}$ could be used instead if compensating by its supremum norm $\ninf{u_{0:T|T}}$.

\begin{proof}
It follows from the assumption of the theorem that $\oQ_t(\tf)(x') = \ninf{\tf \cdot g_t(x',\cdot)}$ for any $\tf \in \boL(\boX_t)$ and
\eqns{
\oP_{0\up T}(\bstf) = \int \sup_{\bsx \in \boX_{0:T}} \bigg( \bstf(\bsx) f(\bsx_0) \prod_{t=1}^T g_t(\bsx_{t-1},\bsx_t) \bigg) P_0(\d f),
}
for any $\bstf \in \boL(\boX_{0:T})$. The associated distribution $P_{0:T}$ on $\boL(\boX_{0:T})$ is the push-forward measure $\zeta_* P_0$ with $\zeta$ the mapping from $\boL(\boX_0)$ to $\boL(\boX_{0:T})$ characterised by
\eqns{
\zeta(f)(\bsx) = f(\bsx_0) \prod_{t=1}^T g_t(\bsx_{t-1},\bsx_t)
}
for any $f \in \boL(\boX_0)$ and any $\bsx \in \boX_{0:T}$. The mapping $\zeta$ is implicitly assumed to be measurable. Since both $\oP_{0:T}$ and $\oR_{0:T}$ take unconditional forms, the posterior distribution $P_{0:T|T}$ which integrates all the observed information can be stated simply as
\eqnsa{
P_{0:T|T}(F) & = \big(P_{0:T} \star (O^*R_{0:T})\big)(F) \\
& \propto \int \ind{F}\big((f \cdot (h\circ O))^{\scl}\big) \ninf{f \cdot (h\circ O)} P_{0:T}(\d f) R_{0:T}(\d h),
}
for any measurable subset $F$ of $\boL(\boX_{0:T})$, where $O = O_0 \times \dots \times O_T : \boX_{0:T} \to \boY_{0:T}$ and where the pointwise product $f \cdot (h\circ O)$ can be expressed for any $f$ and any $h$ in the support of $P_{0:T}$ and $R_{0:T}$ respectively as
\eqns{
(f \cdot (h\circ O))(\bsx) = f_0(\bsx_0) \prod_{t=1}^T g_t(\bsx_{t-1},\bsx_t) \prod_{t=0}^T h_t(O_t(\bsx_t)),
}
for any $\bsx \in \boX_{0:T}$ and for some $f_0 \in \boL(\boX_0)$ and some $\{h_t\}_{t\in\bbT}$ such that $\zeta(f_0) = f$ and $h(\bsy) = h_0(\bsy_0) \dots h_T(\bsy_T)$. The form taken by the smoothed outer measure $\oP_{0:T|T}$ can then be easily deduced.
\end{proof}

\begin{example}
\label{ex:singleFunctionSmoothing}
If the prior knowledge $P_0$ and the observed information $\{R_t\}_{t\in\bbT}$ are further simplified to $P_0 = \delta_{f_0}$ and $R_t = \delta_{h_t}$ for all $t \in \bbT$, then $\oP_{0:T|T}(\bstf) = \ninf{\bstf \cdot f_{0:T|T}}$ for any $\bstf \in \boL^{\infty}(\boX_{0:T})$, with
\eqns{
f_{0:T|T}(\bsx) \propto f_0(\bsx_0) \prod_{t=1}^T g_t(\bsx_{t-1},\bsx_t) \prod_{t=0}^T h_t(O_t(\bsx_t))
}
which has the exact same form as the usual smoothing distribution \cite{Briers2010}
\eqns{
p_{0:T|T}(\bsx) \propto p_0(\bsx_0) \prod_{t=1}^T q_t(\bsx_{t-1}, \bsx_t) \prod_{t=0}^T \ell_t(y_t \given \bsx_t),
}
with $\ell_t(y_t \given x) = h_t(O_t(x))$ the likelihood of a standard observation $y_t$ at time $t \in \bbT$, and where probability density functions are written with the same notation as their corresponding measure. The only difference between these two expressions is that possibility functions replace probability distributions. The expression $p_{0:T|T}$ can also be recovered from $\oP_{0:T|T}$ by assuming that $P_t$ and $Q_t$ are equivalent to the distribution $p_t$ and Markov kernel $q_t$ at each time $t$. This does not however limit the modelling options of the observed information. 
\end{example}

\section{Filtering}
\label{sec:filtering}

The objective is now to compute the information at successive times in a recursive fashion. The predicted and updated filtering outer measures $\oP_{t|t-1}$ and $\oP_{t|t}$ at time $t \in \bbT$ could be simply expressed as the marginals of predicted smoothing outer measure $\oP_{0:t|t-1}$ and the updated smoothing outer measure $\oP_{0:t|t}$, that is as
\eqns{
\oP_{t|t-1}(\tf) = \oP_{0:t|t-1}(\tf) \AND \oP_{t|t}(\tf) = \oP_{0:t|t}(\tf)
}
with $\tf \in \boL^{\infty}(\boX_t)$. However, as in the standard approach, this gives little insight into how to actually compute these terms. Instead, the predicted outer measure $\oP_{t|t-1}$ at time $t$ has to be expressed as a function of the updated outer measure $\oP_{t-1|t-1}$ at the previous time and, similarly, the updated outer measure $\oP_{t|t}$ at time $t$ has to be expressed as a function of the predicted one.

We assume that at a given time $t-1$, $\oP_{t-1|t-1}$ is in the single-variate form
\eqns{
\oP_{t-1|t-1}(\tf) = \int \ninf{\tf \cdot f} P_{t-1|t-1}(\d f),
}
with $\tf \in \boL^{\infty}(\boX_{t-1})$. The predicted outer measure $\oP_{t|t-1}$ is the marginal on $\boX_t$ of the outer measure $\oP_{t-1|t-1} \oQ_t$ on the joint space $\boX_{t-1}\times\boX_t$, which can be expressed as
\eqnmla{eq:generalPrediction}{
\oP_{t|t-1}(\tf) & = \oP_{t-1|t-1} \oQ_t(\tf) \\
& = \int \ninf{ f \cdot \oQ_t(\tf) } P_{t-1|t-1}(\d f),
}
for any $\tf \in \boL^{\infty}(\boX_t)$. As with smoothing, this expression does not reduce to a single-variate outer measure in general so that special cases are considered in the following sections. We proceed as in \cref{sec:smoothing} to obtain a closed-form expression of the filtering equations.

\begin{theorem}
If for any $t \in \{1,\dots,T\}$ there exists a function $g_t(x,\cdot) \in \boL(\boX_t)$ such that $Q_t(x,\cdot) = \delta_{g_t(x,\cdot)}$ for any $x \in \boX_{t-1}$ then the predicted and updated distributions $P_{t|t-1}$ and $P_{t|t}$ are characterised by
\eqnl{eq:filteringUncertainTransition}{
\begin{cases}
P_{t|t-1} = (\xi_t)_* P_{t-1|t-1} \\
P_{t|t} = P_{t|t-1} \star (O_t^* R_t),
\end{cases}
}
where the mapping $\xi_t$ from $\boL(\boX_{t-1})$ to $\boL(\boX_t)$ is characterised by
\eqnsa{
\xi_t(f) : \boX_t & \to [0,1] \\
x & \mapsto \sup_{x' \in \boX_{t-1}} f(x') g_t(x',x).
}
for any $f \in \boL(\boX_{t-1})$.
\end{theorem}

\begin{proof}
With the considered assumption, \eqref{eq:generalPrediction} simplifies to
\eqnmla{eq:predConstraint}{
\oP_{t|t-1}(\tf) & = \int \ninf{ \tf \cdot \xi_t(f') } P_{t-1|t-1}(\d f') \\
& = \int \ninf{\tf \cdot f} (\xi_t)_*P_{t-1|t-1}(\d f).
}
The outer measure $\oP_{t|t-1}$ is now single-variate and the corresponding distribution on $\boL(\boX_t)$ is $P_{t|t-1} \defeq (\xi_t)_* P_{t-1|t-1}$. The next step is to incorporate the observed information $R_t$ in the predicted distribution $P_{t|t-1}$. Since the operation $\star$ defined in \eqref{eq:fusion} can be directly applied to these single-variate distributions, we find that $P_{t|t} = P_{t|t-1} \star (O_t^* R_t)$. To sum up, the filtering equations can be expressed in terms of probability distributions on $\boL(\boX_{t-1})$ and $\boL(\boX_t)$ since all the outer measures involved are single-variate under the considered assumptions.
\end{proof}

\begin{example}
To understand the mapping $\xi_t$, assume that $g_t(x',\cdot) = \ind{G_{x'}}$ for some subset $G_{x'}$ of $\boX_t$, i.e.\ if the considered system is in state $x'$ at time $t-1$ then it is only known that its state at time $t$ is within the subset $G_{x'}$. It follows that
\eqnsa{
\xi_t(\ind{A'})(x) & = \sup_{x' \in \boX_{t-1}} \ind{A'}(x') \ind{G_{x'}}(x) \\
& =
\begin{cases*}
1 & if there exists $x' \in \boX_{t-1}$ s.t.\ $x \in G_{x'}$ and $x' \in A'$ \\
0 & otherwise.
\end{cases*}
}
This can be written as $\xi_t(\ind{A'}) = \ind{A}$ with
\eqns{
A = \bigcup_{x' \in A'} G_{x'}.
}
If $G_{x'}$ is translation invariant, i.e.\ the extent of the set $G_{x'}$ does not depend on $x'$, then, in the language of mathematical morphology, $A$ is a \emph{dilatation} of $A'$ by $G_{x'}$. If $P_{t-1|t-1} = \delta_{\ind{A'}}$ then
\eqns{
P_{t-1|t-1}(\xi_t^{-1}(F)) =
\begin{cases*}
1 & if $\ind{A'} \in \xi_t^{-1}(F)$ \\
0 & otherwise,
\end{cases*}
}
where the condition $\ind{A'} \in \xi_t^{-1}(F)$ is equivalent to $\xi_t(\ind{A'}) = \ind{A} \in F$ so that $(\xi_t)_* P_{t-1|t-1} = \delta_{\ind{A}}$ as expected.
\end{example}

\begin{remark}
If the initial information $P_0$ is equivalent to a probability measure $p_0$ then a particle representation $\{x_i\}_{i=1}^N$ of $p_0$ can be used to approximate $P_0$ as $P_0 \approx N^{-1} \sum_{i=1}^N \delta_{\ind{x_i}}$. The recursion \eqref{eq:filteringUncertainTransition} can then be more easily applied.
\end{remark}

\begin{example}
The filtering equations \eqref{eq:filteringUncertainTransition} imply that if the information provided at time $t$ via $R_t$ takes the form of a probability measure $r_t$ on the state space, then $P_{t|t}$ will also be equivalent to some probability measure $p_{t|t}$ on $\boX_t$. The predicted information $P_{t+1|t}$ will however tend to take a slightly more complicated form: it will give probability mass $p_{t|t}(\d x')$ to the function $g_{t+1}(x',\cdot)$ on $\boX_{t+1}$. If at time $t+1$, the observation $R_{t+1}$ is once again equivalent to a probability measure $r_{t+1}$ on $\boX_{t+1}$, then the distribution $P_{t+1|t+1} = P_{t+1|t} \star R_{t+1}$ will be of the form
\eqns{
P_{t+1|t+1}(F) \propto \int \ind{F}(\ind{x}) \ninf{ \ind{x} \cdot g_{t+1}(x',\cdot) } r_{t+1}(\d x) p_{t|t}(\d x'),
}
that is $P_{t+1|t+1}$ will be equivalent to a probability measure on $\boX_{t+1}$. If, additionally, $p_{t|t}$ and $r_t$ are Gaussian distributions and $g_{t+1}(x',x)$ is the Gaussian possibility function $\bar\calN(x; \bsF_{t+1} x',\bsQ_{t+1})$ for some matrices $\bsF_{t+1}$ and $\bsQ_{t+1}$ then $P_{t+1|t+1}$ is equivalent to the corresponding posterior Gaussian distribution of the Kalman filter.
\end{example}

\section{Special cases and related results}

We first detail two special cases of the approach introduced in \cref{sec:filtering} where the filtering recursion is expressed without measure-theoretic notations by reducing the functional integrals to finite sums. The second case restricts the system to be linear and based on Gaussian possibility functions. 

\subsection{Filtering with finite sum of possibility functions}

Assume that the predicted distribution $P_{t-1|t-1}$ and the observed information $R_t$ take the form of a finite sum of functions as follows:
\eqns{
P_{t-1|t-1} = \sum_{i \in I_{t-1}} w^i_{t-1} \delta_{f_{t-1|t-1}^i} \AND R_t = \sum_{l \in L_t} v^l_t \delta_{h_t^l}
}
for some indexed families $\{ (w^i_{t-1},f_{t-1|t-1}^i) \}_{i \in I_{t-1}}$ and $\{(v^l_t, h_t^l)\}_{l \in L_t}$ of pairs of weights and functions, then the predicted and updated distributions $P_{t|t-1}$ and $P_{t|t}$ can be expressed as
\eqns{
P_{t|t-1} = \sum_{i \in I_{t-1}} w^i_{t|t-1} \delta_{f_{t|t-1}^i} \AND P_{t|t} = \sum_{i \in I_t} w^i_t \delta_{f^i_{t|t}}
}
where $I_t = I_{t-1} \times L_t$ and where
\eqns{
\begin{cases}
\displaystyle \big( w^i_{t|t-1}, f^i_{t|t-1} \big) = \bigg( w^i_{t-1},\quad \sup_{x' \in \boX_{t-1}} f^i_{t-1|t-1}(x') g_t(x',\cdot) \bigg) \\
\hspace{7cm}\text{for any $i \in I_{t-1}$} \\
 \\
\displaystyle \big( w^i_t, f^i_{t|t} \big) = \bigg( \dfrac{ \ninf{f^{j,l}_{t,t-1}} w^j_{t-1} v^l_t }{ \sum_{(k,m) \in I_t} \ninf{ f^{j,l}_{t,t-1} } w^k_{t-1} v^m_t },\quad (f^{j,l}_{t,t-1})^{\scl}\bigg) \\
\hspace{7cm}\text{for any $i = (j,l) \in I_t$},
\end{cases}
}
with $f^{j,l}_{t,t-1} = f^j_{t|t-1} \cdot (h^l_t \circ O_t)$ for any $(j,l) \in I_t$.

This recursion could be easily computed if the considered possibility functions are part of a parametric family of functions such as indicator functions or Gaussian possibility functions.

In the simplest case where $P_{t-1|t-1}=\delta_{f_{t-1|t-1}}$ and $R_t = \delta_{h_t}$, the filtering equations can be expressed in standard notations as
\eqnl{eq:filteringPossibility}{
\begin{cases}
\displaystyle f_{t|t-1}(x) = \sup_{x' \in \boX_{t-1}} f_{t-1|t-1}(x') g_t(x',x) \\
\vspace{-0.75em}\\
\displaystyle f_{t|t}(x) = \dfrac{ f_{t|t-1}(x) h_t(O_t(x)) }{ \sup_{x' \in \boX_t} f_{t|t-1}(x') h_t(O_t(x'))}.
\end{cases}
}
As in \cref{ex:singleFunctionSmoothing}, these filtering equations are similar to the ones of the standard formulation but with integrals replaced by supremums and distributions replaced with possibility functions. It is interesting to study \eqref{eq:filteringPossibility} under Kalman-like assumptions of Gaussianity and linearity as in the following section.

\subsection{Filtering for linear system with Gaussian possibility function}

A natural question that arises from the simple form of the filtering equations \eqref{eq:filteringPossibility} is: how would such a recursion perform under assumptions of linearity and when only Gaussian possibility functions are involved? Since the information that is given to the algorithm is weaker when compared to the one given to the standard Kalman filter \cite{Anderson1979}, one might expect that the algorithm based on possibilities will be more robust to modelling discrepancies. However, it might also be expected to be less accurate than the standard Kalman filter when dynamics and observation are indeed generated according to the assumed Gaussian distributions. The following theorem shows that, interestingly, both algorithms are equivalent when characterised by their respective means and variance/spread.

\begin{theorem}
\label{thm:Kalman}
Assume that the transition function $F_t( \cdot , V_t)$ and the observation function $O_t$ are linear. Also, assume that the noise $V_t$ is additive and described by a Gaussian possibility function. If the prior possibility function $f_{t-1|t-1}$ and the observed-information $h_t$ are Gaussian, then the mean and spread of the possibility functions in \eqref{eq:filteringPossibility} follow the standard Kalman filter recursion.
\end{theorem}

\begin{proof}
The assumptions on the transition, observation and prior possibility function can be expressed as
\eqnsa{
f_{t-1|t-1}(x') & = \bar\calN(x'; m_{t-1},\bsP_{t-1}) \\
g_t(x',x) & = \bar\calN(x; \bsF_t x', \bsQ_t) \\
h_t(y) & = \bar\calN(y_t; y, \bsR_t) \\
O_t(x) & = \bsO_t x,
}
for some $y_t \in \boY_t$ representing the observation in the usual way, some $m_{t-1} \in \boX_{t-1}$ and some matrices $\bsP_{t-1}$, $\bsF_t$, $\bsQ_t$, $\bsR_t$ and $\bsO_t$ of appropriate size. Using an equivalent formulation to the standard Kalman filter identity (easily obtained by Sylvester's determinant theorem), expressed as
\eqns{
\bar\calN(x; \bsF x', \bsQ) \bar\calN(x';m, \bsP) = \bar\calN\big(x; \bsF m, \bsQ+\bsF\bsP\bsF^T \big)
\bar\calN\big(x'; m + \bsK(x - \bsF m), (\bsI-\bsK\bsF)\bsP\big)
}
with $\bsK = \bsP\bsF^T(\bsF\bsP\bsF^T + \bsQ)^{-1}$, it follows that
\eqns{
f_{t|t-1}(x) = \bar\calN\big(x; \bsF_t m_{t-1}, \bsQ_t+\bsF_t\bsP_{t-1}\bsF_t^T \big)
\sup_{x' \in \boX_{t-1}} \bar\calN\big(x'; m', \bsP' \big),
}
for some state $m'$ and some matrix $\bsP'$. The supremum in the r.h.s.\ of this expression is equal to $1$ so that the Kalman filter time-prediction is recovered:
\eqnsa{
f_{t|t-1}(x) & = \bar\calN(x; m_{t|t-1}, \bsP_{t|t-1}) \\
& \defeq \bar\calN\big(x; \bsF_t m_{t-1}, \bsQ_t+\bsF_t\bsP_{t-1}\bsF_t^T \big).
}
The Kalman-filter observation update can be recovered in a similar fashion and is found to be
\eqns{
f_{t|t}(x) = \bar\calN\big(x; m_{t|t-1} + \bsK_t(y_t - \bsO_t m_{t|t-1}), (\bsI-\bsK_t \bsO_t)\bsP_{t|t-1}\big)
}
with $\bsK_t = \bsP_{t|t-1}\bsO_t^T(\bsO_t\bsP_{t|t-1}\bsO_t^T + \bsR)^{-1}$. The Kalman filter recursions are then recovered in spite of the presence of supremums instead of integrals in \eqref{eq:filteringPossibility}.
\end{proof}

A related result, named the Kalman evidential filter \cite{Mahler2007}, has been proved in the context of fuzzy Dempster-Shafer theory with a fully-known Markov transition. This assumption however does not allow for recovering the Kalman filter exactly, but yields a recursive algorithm that bears some similarities.

\subsection{Backward smoothing recursion}

Obtaining the expression of the distribution $P_{t|T}$ on $\boL(\boX_t)$ representing the uncertainty at time $t$ given all the observed information up to time $T$ is useful for recovering the smoothed distribution after one filtering pass on the set $\bbT$ of all time steps.

\begin{theorem}
\label{thm:backwardSampling}
If for any $t \in \{1,\dots,T\}$ there exists $g_t(x,\cdot) \in \boL(\boX_t)$ such that $Q_t(x,\cdot) = \delta_{g_t(x,\cdot)}$ for any $x \in \boX_{t-1}$, then the smoothed outer measure $\oP_{0:T|T}$ can be expressed as
\eqnl{eq:backwardSampling}{
\oP_{T\down 0|T}(\bstf) = \oP_{T|T} \oP'_{T-1|T-1} \dots \oP'_{0|0}(\bstf)
}
for any $\bstf \in \boL^{\infty}(\boX_{0:T})$, where the conditional outer measure $\oP'_{t|t}(\tf)(x_{t+1})$ is defined for any $\tf \in \boL^{\infty}(\boX_t)$ and any $x_{t+1} \in \boX_{t+1}$ as
\eqnl{eq:backwardSampling_t}{
\oP'_{t|t}(\tf)(x_{t+1}) = \dfrac{\oP_{t|t}(\tf \cdot g_{t+1}(\cdot,x_{t+1}))}{ \oP_{t|t}(g_{t+1}(\cdot,x_{t+1})) }.
}
\end{theorem}

\begin{proof}
The probability distribution $P_{0:T|T}$ on $\boL(\boX_{0:T})$ defined in \cref{thm:smoothing} is supported by possibility functions of the form
\eqnsa{
f_{0:T|T}(\bsx) & = \dfrac{f_0(\bsx_0) \prod_{t=1}^T g_t(\bsx_{t-1},\bsx_t) \prod_{t=0}^T h_t(O_t(\bsx_t))}{ \sup_{\bsx'} f_0(\bsx'_0) \prod_{t=1}^T g_t(\bsx'_{t-1},\bsx_t) \prod_{t=0}^T h_t(O_t(\bsx'_t)) } \\
& = f_{T|T}(\bsx_T) \prod_{t = 0}^{T-1} f'_{t|t}( \bsx_t \given \bsx_{t+1} )
}
with
\eqns{
f'_{t|t}( x_t \given x_{t+1}) = \dfrac{ g_{t+1}(x_t,x_{t+1}) f_{t|t}(x_t) }{ \sup_x g_t(x,x_{t+1}) f_{t|t}(x) }.
}
The outer measure $\oP_{T\down 0|T}$ can then expressed as in \eqref{eq:backwardSampling} where, for any $t \in \{0,\dots,T-1\}$, $\oP'_{t|t}$ is induced by a distribution $P'_{t|t}(\cdot \given x_{t+1})$ supported by possibility functions of the form $f'_{t|t}(\cdot \given x_{t+1})$ for any $x_{t+1} \in \boX_{t+1}$. However, the possibility function $f'_{t|t}$ can be recognised as the one yielded by the combination of $Q_{t+1}$ and $P_{t|t}$, which implies that $\oP'_{t|t}$ can be equally expressed as \eqref{eq:backwardSampling_t}, hence proving the theorem.
\end{proof}

\subsection{Filtering with known transition}

The general recursion \eqref{eq:generalPrediction} can also be made closed-form by following an approach that is the exact opposite of the one considered in \cref{sec:filtering}, that is by making the Markov kernel extremely informative: it is assumed that for any $x \in \boX_{t-1}$ the transition $Q_t(x,\cdot)$ is equivalent to a Markov kernel $q_t(x,\cdot)$ from $\boX_{t-1}$ to $\boX_t$, that is $Q_t(x,\cdot)$ gives mass $q_t(x,\d x')$ to the degenerate possibility function $\ind{x'}$. From this assumption, the outer measure $\oQ_t$ verifies
\eqns{
\oQ(\ind{B})(x) = \int \ind{B}(x') q_t(x, \d x') = q_t(x,B),
}
for any $B \in \calB(\boX_t)$. However, this assumption is not sufficient for \eqref{eq:generalPrediction} to simplify unless $P_{t-1|t-1}$ is also equivalent to a probability measure $p_{t-1|t-1}$ on $\boX_{t-1}$, in which case \eqref{eq:generalPrediction} becomes the standard time prediction
\eqnl{eq:stdPrediction}{
p_{t|t-1}(B) = \int q_t(x,B) p_{t-1|t-1}(\d x),
}
for any $B \in \calB(\boX_t)$. The update requires $R_t$ to be restricted to non-random uncertainty, otherwise the observed information would be incompatible with the predicted information. The observation update becomes
\eqns{
p_{t|t}(B) = \dfrac{\int \ind{B}(x) h(x) R_t(\d h) p_{t|t-1}(\d x)}{\int h(x) R_t(\d h) p_{t|t-1}(\d x)},
}
for any $B \in \calB(\boX_t)$, so that both the predicted and the updated uncertainties take the form of probability measures on the state space and only the observed information takes a more general form. This approach has been previously proposed in the context of random set theory \cite{Mahler2007} and has also been used for multi-target tracking within the proposed framework in \cite{Houssineau2016_HISP,Delande2017}.

As an example, if we assume that $R_t$ takes the form $R_t = \sum_{l \in L_t} v^l_t \delta_{h_t^l}$ for some index set $L_t$ and some collections of weights $\{v^l_t\}_{l \in L_t}$ and functions $\{h_t^l\}_{l \in L_t}$, then it holds that
\eqns{
p_t(\d x) = \dfrac{\sum_{l \in L_t} v^l_t h^l_t(x) p_{t|t-1}(\d x)}{\sum_{l \in L_t} v^l_t \int h^l_t(x) p_{t|t-1}(\d x)}.
}
The recursion based on the standard prediction \eqref{eq:stdPrediction} and this update can be computed using sequential Monte Carlo methods where the likelihood is replaced by a potential $\sum_{l\in L_t} v^l_t h^l_t$.

\section{Concluding Remarks}

Building on a recently introduced framework for the representation of uncertainty \cite{Houssineau2015,Houssineau2016_dataAssimilation}, it has been demonstrated that filtering and smoothing algorithms can be generalised to outer measures belonging to a specific class based on functional integrals of supremums. An important observation was that the structure of the usual filtering and smoothing equations reappears in the generalised recursion under the form of possibility functions upon which the outer measures were defined. Simplifications to finite sums of possibility functions as well as single possibility functions have been studied and gave results that were not only intuitive but also implementable. The recursion in terms of mean and spread of the Kalman filter has been recovered by considering appropriately defined Gaussian possibility functions, giving yet another setting in which the Kalman filter appears naturally. 

Future work may include the application of Monte Carlo-like methods to the proposed estimation framework. The results obtained in this article raise numerous other questions, both of a practical and theoretical nature. For example:
\begin{enumerate}
\item Given the result of \cref{thm:Kalman}, it is natural to inquire about inference in the non-Gaussian case. In the standard approach, the most straightforward generalisation is based on Gaussian mixtures \cite{Sorenson1971,Alspach1972} of the form
\eqns{
\sum_{i=1}^N \tilde{w}_i \calN(x; m_i, \bsP_i)
}
for some integer $N$ and some collections $\{\tilde{w}_i\}_{i=1}^N$, $\{m_i\}_{i=1}^N$, $\{\bsP_i\}_{i=1}^N$ of $[0,1]$-valued scalars, states in $\bbR^d$ and $d\times d$ positive definite matrices respectively. In particular, it holds that $\sum_{i=1}^N \tilde{w}_i = 1$. In the considered framework, mixtures become \emph{max-mixtures} and take the form
\eqns{
\max_{1 \leq i \leq N} w_i \bar\calN(x; m_i, \bsP_i)
}
with $\{w_i\}_{i=1}^N$ a collection of $[0,1]$-valued scalars such that $\max_{1\leq i\leq N} w_i = 1$. Inference for these max-mixtures requires adequate mixture reduction techniques.
\item As mentioned shortly after \cref{def:GaussianPossibility}, the parameters $m$ and $\bsP$ in the Gaussian possibility function $\bar\calN(\cdot; m, \bsP)$ are referred to as mean and spread only as a useful abuse of language. It would however be important, both from the theoretical and practical viewpoints, to formally introduce these concepts. In particular, the law of large numbers and central limit theorem, assuming they can be reformulated to suit outer measures, would provide insight and theoretical backup for a meaningful generalised definition of the concepts of mean and variance.
\end{enumerate}

\section*{Acknowledgements}

J.~Houssineau was with Data61 (CSIRO) and is now with the National University of Singapore. A.N.~Bishop is with the University of Technology Sydney (UTS) and Data61 (CSIRO). He is also an adjunct Fellow at the Australian National University. He is supported by Data61 and the Australian Research Council (ARC) via a Discovery Early Career Researcher Award (DE-120102873). His work was partly supported by DST Group under TTCP CREATE (2017).

\bibliographystyle{siamplain}
\bibliography{Uncertainty}

\end{document}